\documentclass[11pt]{article}
\usepackage[usenames,dvipsnames]{color}
\usepackage{amssymb,amsmath,graphics,color,enumerate}
\usepackage{tikz}
\usepackage{xspace}
\usetikzlibrary{decorations.pathreplacing,decorations.pathmorphing,patterns,calc,arrows}
\usepackage{todonotes}
\usepackage[utf8x,utf8]{inputenc} 
\usepackage[T1]{fontenc}
\usepackage[vlined,linesnumbered,boxruled]{algorithm2e}
\usepackage[left=1in,right=1in,top=1in,bottom=1in]{geometry}

\usepackage{comment}

\usepackage{epsfig}
\usepackage{subfig}
\usepackage{float}
\usepackage{amsmath,amsfonts,amssymb,epsfig,color,amsthm}
\usepackage{url}

\usepackage{thmtools, thm-restate}
\declaretheorem{theorem}

\newtheorem{lemma}[theorem]{Lemma}
\newtheorem{corollary}[theorem]{Corollary}
\newtheorem{proposition}[theorem]{Proposition}

\newtheorem{definition}[theorem]{Definition}
\newcommand{\ignore}[1]{}%

\newcommand{\ProblemFormat}[1]{{\sc #1}}
\newcommand{\ProblemName}[1]{\ProblemFormat{#1}\xspace}

\newcommand{\probLab}[1][($a$:$b)$]{\ProblemName{List #1-coloring}}
\newcommand{\probNULab}[1][($a$:$b)$]{\ProblemName{Nonuniform List #1-coloring}}
\newcommand{\probABCol}[1][($a$:$b$)]{\ProblemName{#1-coloring}}
\newcommand{\probHom}{\ProblemName{Graph Homomorphism}}
\newcommand{\probTSAT}{\ProblemName{$3$-SAT}}
\newcommand{\probTFSAT}{\ProblemName{$(3,\!4)$-SAT}}

\newcommand{\probMonoTest}{\ProblemName{$(r,k)$-Monomial Testing}}
\newcommand{\probCLSubsetSum}{\ProblemName{Carry-Less Subset Sum}}

\newcommand{\classP}{{\ensuremath{\rm{P}}}}
\newcommand{\NP}{{\ensuremath{\rm{NP}}}}

\DeclareMathOperator{\poly}{poly}

\newcounter{rulecnt}

\newcommand{\field}[1]{\textup{GF}(#1)}

\newcommand{\Ff}{{\ensuremath{\mathcal{F}}}}

\newcommand{\Hh}{{\ensuremath{\mathcal{H}}}}
\newcommand{\Ii}{{\ensuremath{\mathcal{I}}}}

\newcommand{\Oh}{\ensuremath{\mathcal{O}}}
\newcommand{\Ohstar}{\ensuremath{\Oh^\star}}

\begin{document}
\title{Tight lower bounds for the complexity of 
\textcolor{blue}{m}\textcolor{red}{u}\textcolor{olive}{l}\textcolor{violet}{t}i\textcolor{pink}{c}\textcolor{violet}{o}\textcolor{orange}{l}\textcolor{green}{o}\textcolor{blue}{r}\textcolor{brown}{i}\textcolor{yellow}{n}\textcolor{red}{g}\thanks{Work supported by the National Science Centre of Poland, grants number 2013/11/D/ST6/03073 (MP, MW) and 2015/17/N/ST6/01224 (AS). The work of {\L}. Kowalik is a part of the project TOTAL that has received funding from the European Research Council (ERC) under the European Union’s Horizon 2020 research and innovation programme (grant agreement No 677651). Micha\l{} Pilipczuk is supported by the Foundation for Polish Science (FNP) via the START stipend programme.}} 


\author{Marthe Bonamy\thanks{CNRS, LaBRI, France} \and \L ukasz Kowalik\thanks{University of Warsaw, Poland} \and Micha\l \ Pilipczuk\footnotemark[3] \and Arkadiusz Soca\l a\footnotemark[3] \and Marcin Wrochna\footnotemark[3]}


\date{}

\begin{titlepage}
\maketitle

\begin{abstract}
In the {\em{multicoloring}} problem, also known as {\em{($a$:$b$)-coloring}} or {\em{$b$-fold coloring}}, we are given a graph $G$ and a set of $a$ colors, 
and the task is to assign a subset of $b$ colors to each vertex of $G$ so that adjacent vertices receive disjoint color subsets.
This natural generalization of the classic coloring problem (the $b=1$ case) is equivalent to finding a homomorphism to the Kneser graph $KG_{a,b}$, and gives relaxations approaching the fractional chromatic number.

We study the complexity of determining whether a graph has an ($a$:$b$)-coloring.
Our main result is that this problem does not admit an algorithm with running time $f(b)\cdot 2^{o(\log b)\cdot n}$, for any computable $f(b)$, unless the Exponential Time Hypothesis (ETH) fails.
A $(b+1)^n\cdot \poly(n)$-time algorithm due to Nederlof [2008] shows that this is tight.
A direct corollary of our result is that the graph homomorphism problem does not admit a $2^{\Oh(n+h)}$ algorithm unless ETH fails, even if the target graph is required to be a Kneser graph. This refines the understanding given by the recent lower bound of Cygan et al. [SODA 2016].

The crucial ingredient in our hardness reduction is
the usage of {\em{detecting matrices}} of Lindstr\"om~[Canad. Math. Bull., 1965], which is a combinatorial tool that, to the best of our knowledge, has not yet been used for proving complexity lower bounds.
As a side result, we prove that the running time of the algorithms of Abasi et al.~[MFCS 2014] and of Gabizon et al.~[ESA 2015] for the $r$-monomial detection problem are optimal under ETH.

\end{abstract}
\pagenumbering{gobble}
\end{titlepage}
\pagenumbering{arabic}

\section{Introduction}\label{sect:intro}

The complexity of determining the chromatic number of a graph is undoubtedly among the most intensively studied computational problems.
Countless variants, extensions, and generalizations of graph colorings have been introduced and investigated.
Here, we focus on {\em multicolorings}, also known as ($a$:$b$)-colorings. In this setting, we are given a graph $G$, a palette of $a$ colors, and a number $b\leq a$.
An {\em ($a$:$b$)-coloring} of $G$ is any assignment of $b$ distinct colors to each vertex so that adjacent vertices receive disjoint subsets of colors.
The \probABCol problem asks whether $G$ admits an ($a$:$b$)-coloring.
Note that for $b=1$ we obtain the classic graph coloring problem. 
The smallest $a$ for which an ($a$:$b$)-coloring exists, is called the {\em{$b$-fold chromatic number}}, denoted by $\chi_b(G)$.

The motivation behind ($a$:$b$)-colorings can be perhaps best explained by showing the connection with the {\em{fractional chromatic number}}.
The fractional chromatic number of a graph $G$, denoted $\chi_f(G)$, is the optimum value of the natural LP relaxation of the problem of computing the chromatic number of $G$,
expressed as finding a cover of the vertex set using the minimum possible number of independent sets.
It can be easily seen that by relaxing the standard coloring problem by allowing $b$ times more colors while requiring that every vertex receives $b$ colors
and adjacent vertices receive disjoint subsets, with increasing $b$ we approximate the fractional chromatic number better and better. Consequently, $\lim_{b\to\infty} \chi_b(G)/b=\chi_f(G)$.

\begin{figure}[h]
	\centering
	\begin{tikzpicture}[scale=1,line width=0.3mm]
\begin{scope}[>=latex]
 
 \foreach \i/\l/\c in {1/v/51,2/x/34,3/y/12,4/u/45,5/w/23}
 {
   \path (-90 + \i * 360 / 5 :8mm) node[draw, circle, minimum size = 0.5cm] (\l1) {};
   \draw (\l1) node[black] {\c};
 }  

 \foreach \i/\l/\c in {1/v/24,2/x/25,3/y/35,4/u/13,5/w/14}
 {
   \path (-90 + \i * 360 / 5 :17mm) node[draw, circle, minimum size = 0.5cm] (\l2){};
   \draw (\l2) node[black] {\c};
 }  

 \foreach \i/\l/\c in {1/v/24,2/x/25,3/y/35,4/u/13,5/w/14}
 {
   \path (90 + \i * 360 / 5 :22mm) node[draw, circle, minimum size = 0.5cm] (\l3){};
   \draw (\l3) node[black] {\c};
 }  

 \foreach \i/\l/\c in {1/v/51,2/x/34,3/y/12,4/u/45,5/w/23}
 {
   \path (90 + \i * 360 / 5 :35mm) node[draw, circle, minimum size = 0.5cm] (\l4){};
   \draw (\l4) node[black] {\c};
 }  

 \foreach \i/\j in {v1/x1,x1/y1,y1/u1,u1/w1,w1/v1,v1/v2,x1/x2,y1/y2,u1/u2,w1/w2,v4/x4,x4/y4,y4/u4,u4/w4,w4/v4,v4/v3,x4/x3,y4/y3,u4/u3,w4/w3,w3/x2,x2/u3,u3/v2,v2/y3,y3/w2,w2/x3,x3/u2,u2/v3,v3/y2,y2/w3}
 {
       \draw (\i) -- (\j);
 }

 \draw[-angle 90] (4,0) -- (5,0);
\end{scope}

\begin{scope}[shift={(8,0)}]
 
 \foreach \i/\l/\c in {1/v/24,2/x/25,3/y/35,4/u/13,5/w/14}
 {
   \path (90 + \i * 360 / 5 :10mm) node[draw, circle, minimum size = 0.5cm] (\l3){};
   \draw (\l3) node[black] {\c};
 }  

 \foreach \i/\l/\c in {1/v/51,2/x/34,3/y/12,4/u/45,5/w/23}
 {
   \path (90 + \i * 360 / 5 :20mm) node[draw, circle, minimum size = 0.5cm] (\l4){};
   \draw (\l4) node[black] {\c};
 }  

 \foreach \i/\j in {v4/x4,x4/y4,y4/u4,u4/w4,w4/v4,v3/v4,x3/x4,y3/y4,u3/u4,w3/w4,w3/x3,x3/u3,u3/v3,v3/y3,y3/w3}
 {
       \draw (\i) -- (\j);
 }
\end{scope}

\end{tikzpicture}
	\caption{A $(5$:$2)$-coloring of the dodecahedron (left) which can be seen as a homomorphism to $KG_{5,2}$ (the Petersen graph, right). 
	The homomorphism is given by identifying the pairs of opposite vertices in the corresponding regular solid.}
	\label{fig:homomorphism}
\end{figure}
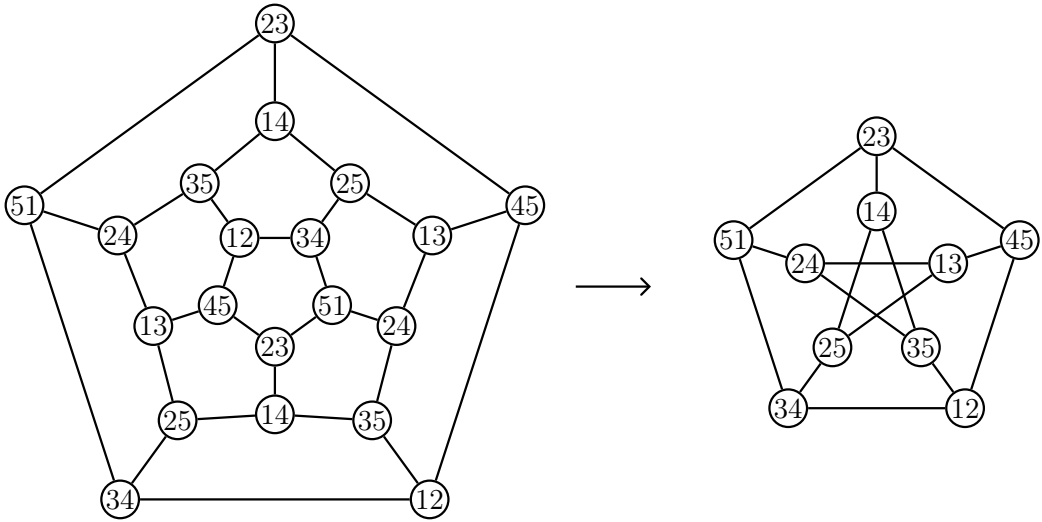

Another interesting connection concerns {\em{Kneser graphs}}. 
Recall that for positive integers $a$, $b$ with $b<a/2$, the Kneser graph $KG_{a,b}$ has all $b$-element subsets of $\{1,2,\ldots,a\}$ as vertices, and two subsets are considered adjacent if and only if they are disjoint. 
For instance, $KG_{5,2}$ is the well-known Petersen graph (see Fig.~\ref{fig:homomorphism}, right). 
Thus, ($a$:$b$)-coloring of a graph $G$ can be interpreted as a homomorphism from $G$ to the Kneser graph $KG_{a,b}$ (see Fig.~\ref{fig:homomorphism}). 
Kneser graphs are well studied in the context of graph colorings mostly due to the celebrated result of Lov\'asz~\cite{Lovasz78}, who determined their chromatic number, initiating the field of topological combinatorics.

Multicolorings and ($a$:$b$)-colorings have been studied both from combinatorial~\cite{chvatal,Fisher95,Lin08} and algorithmic~\cite{ChristFL14,HalldorssonKPSST03,Havet01,KchikechT06,Kuhn09,Marx02,McDiarmidR00,SudeepV05} 
points of view. The main real-life motivation comes from the problem of assigning frequencies to nodes in a cellular network so that adjacent nodes receive disjoint sets of frequencies on which they can operate.
This makes (near-)planar and distributed settings particularly interesting for practical applications. We refer to the survey of Halld\'orsson and Kortsarz~\cite{HalldorssonK04} for a broader discussion.

In this paper we focus on the paradigm of exact exponential time algorithms: given a graph $G$ on $n$ vertices and numbers $a\geq b$, we would like to determine whether $G$ is ($a$:$b$)-colorable as quickly as possible.
Since the problem is already NP-hard for $a=3$ and $b=1$, we do not expect it to be solvable in polynomial time, and hence look for an efficient exponential-time algorithm.
A straightforward dynamic programming approach yields an algorithm with running time\footnote{The $\Ohstar(\cdot)$ notation hides factors polynomial in the input size.}  $\Ohstar(2^n\cdot (b+1)^n)$ as follows. For each function $\eta\colon V(G)\to \{0,1,\ldots,b\}$ and each $k=0,1,\ldots,a$, we
create one boolean entry $D[\eta,k]$ denoting whether one can choose $k$ independent sets in $G$ so that every vertex $v\in V(G)$ is covered exactly $\eta(v)$ times. Then
value $D[\eta,k]$ can be computed as a disjunction of values $D[\eta',k-1]$ over $\eta'$ obtained from $\eta$ by subtracting $1$ on vertices from some independent set in $G$.

This simple algorithm can be improved by finding an appropriate algebraic formula for the number of ($a$:$b$)-colorings of the graph and using the inclusion-exclusion principle to compute it quickly, similarly
as in the case of standard colorings~\cite{BjorklundHK09}. Such an algebraic formula was given by Nederlof~\cite[Theorem 3.5]{Nederlof} in the context of a more general {\sc{Multi Set Cover}} problem.
Nederlof also observed that in the case of \probABCol, a simple application of the inclusion-exclusion principle to compute the formula yields an $\Ohstar((b+1)^n)$-time exponential-space algorithm.
Hua et al.~\cite{HuaWYL10} noted that the formulation of Nederlof~\cite{Nederlof} for {\sc{Multi Set Cover}} can be also used to obtain a polynomial-space algorithm for this problem.
By taking all maximal independent sets to be the family in the {\sc{Multi Set Cover}} problem, and applying the classic Moon-Moser upper bound on their number~\cite{Moon1965}, we obtain an algorithm for \probABCol
that runs in time $\Ohstar(3^{n/3}\cdot (b+1)^n)$ and uses polynomial space. Note that by plugging $b=1$ to the results above, we obtain algorithms for the standard coloring problem 
with running time $\Ohstar(2^n)$ and exponential space usage, and with running time $\Ohstar(2.8845^n)$ and polynomial space usage, which almost matches the fastest known procedures~\cite{BjorklundHK09}.

The complexity of \probABCol becomes particularly interesting in the context of the \probHom problem: given graphs $G$ and $H$, with $n$ and $h$ vertices respectively, determine whether $G$ admits a homomorphism to $H$.
By the celebrated result of Hell and Ne\v{s}et\v{r}il~\cite{HellN90} the problem is in $\classP$ if $H$ is bipartite and $\NP$-complete otherwise.
For quite a while it was open whether there is an algorithm for \probHom running in time $2^{\Oh(n+h)}$. 
It was recently answered in the negative by Cygan et al.~\cite{cygan:homo}; more precisely, they proved that an algorithm with running time $2^{o(n\log h)}$ would contradict the Exponential Time Hypothesis (ETH) of Impagliazzo et al.~\cite{eth}.
However, \probHom is a very general problem, hence researchers try to uncover a more fine-grained picture and identify families of graphs $\Hh$ such that the problem can be solved more efficiently whenever $H\in \Hh$.
For example, Fomin, Heggernes and Kratsch~\cite{DBLP:journals/mst/FominHK07} showed that when $H$ is of treewidth at most $t$, then \probHom can be solved in time $\Ohstar((t+3)^n)$.
It was later extended to graphs of cliquewidth bounded by $t$, with an $\Ohstar((2t+1)^{\max\{n,h\}})$ time bound by Wahlstr{\"{o}}m~\cite{Wahlstrom11}.
On the other hand, $H$ needs not be sparse to admit efficient homomorphism testing: the family of cliques admits the $\Ohstar(2^n)$ running time as shown by Bj\"{o}rklund et al.~\cite{BjorklundHK09}. 
As noted above, this generalizes to Kneser graphs $KG_{a,b}$, by the $\Ohstar((b+1)^n)$-time algorithm of Nederlof.
In this context, the natural question is whether the appearance of $b$ in the base of the exponent is necessary, or is there an algorithm running in time $\Ohstar(c^n)$ for some universal constant $c$ independent of $b$.	

\paragraph*{Our contribution.} We show that the algorithms for \probABCol mentioned above are essentially optimal under the Exponential Time Hypothesis. Specifically, we prove the following results:


\begin{restatable}{theorem}{mainthm}\label{th:main}
If there is an algorithm for \probABCol that runs in time~$f(b) \cdot 2^{o(\log b) \cdot n}$, for some computable function $f(b)$, then ETH fails.
This holds even if the algorithm is only required to work on instances where $a = \Theta(b^2 \log b)$.
\end{restatable}
\vspace*{-\baselineskip}
\begin{restatable}{corollary}{maincor}\label{cor:main}
If there is an algorithm for \probHom that runs in time~$f(h) \cdot 2^{o(\log \log h) \cdot n}$, for some computable function $f(h)$, then ETH fails.
This holds even if the algorithm is only required to work on instances where $H$ is a Kneser graph $KG_{a,b}$ with $a = \Theta(b^2 \log b)$.
\end{restatable}

%

The bound for \probABCol is tight, as the straightforward $\Ohstar(2^n \cdot (b+1)^n)=2^{\Oh(\log b)\cdot n}$ dynamic programming algorithm already shows.
At first glance, one might have suspected that \probABCol, as an interpolation between classical coloring and fractional coloring, both solvable in $2^{\Oh(n)}$ time~\cite{GrotschelLS81}, should be just as easy; Theorem 1 refutes this suspicion.

Corollary~\ref{cor:main} in particular excludes any algorithm for testing homomorphisms into Kneser graphs with running time $2^{\Oh(n+h)}$.
It cannot give a tight lower bound matching the result of Cygan et~al.~\cite{cygan:homo} for general homomorphisms, because $h=|V(KG_{a,b})|=\binom{a}{b}$ is not polynomial in $b$.
On the other hand, it exhibits the first explicit family of graphs $H$ for which the complexity of \probHom increases with $h$.

In our proof, we first show a lower bound for the list variant of the problem, where every vertex is given a list of colors that can be assigned to it (see Section~\ref{sect:pre} for formal definitions).
The list version is reduced to the standard version by introducing a large Kneser graph $KG_{a+b,b}$; we need $a$ and $b$ to be really small so that the size of this Kneser graph does not dwarf the size of the rest
of the construction.
However, this is not necessary for the list version, where we obtain lower bounds for a much wider range of functions $b(n)$.


\begin{restatable}{theorem}{listabthm}
\label{th:listab}%
If there is an algorithm for \probLab that runs in time~$2^{o(\log b) \cdot n}$, then ETH fails.
This holds even if the algorithm is only required to work on instances where $a=\Theta(b^2\log b)$ and $b=\Theta(b(n))$ for an arbitrarily chosen polynomial-time computable function $b(n)$ such that $b(n)\in\omega(1)$ and $b(n)=\Oh(n/\log n)$.
\end{restatable}

%
The crucial ingredient in the proof of Theorem~\ref{th:listab} is the usage of {\em{$d$-detecting matrices}} introduced by Lindstr\"om~\cite{lindstrom1965combinatorial}.
We choose to work with their combinatorial formulation, hence we shall talk about {\em{$d$-detecting families}}.
Suppose we are given some universe $U$ and there is an unknown function $f\colon U\to \{0,1,\ldots,d-1\}$, for some fixed positive integer $d$.
One may think of $U$ as consisting of coins of unknown weights that are integers between $0$ and $d-1$.
We would like to learn $f$ (the weight of every coin) by asking a small number of queries of the following form: for a subset $X\subseteq U$, what is $\sum_{e\in X} f(e)$ (the total weight of coins in $X$)?
A set of queries sufficient for determining all the values of an arbitrary $f$ is called a {\em{$d$-detecting family}}. 
Of course $f$ can be learned by asking $|U|$ questions about single coins, but it turns out that significantly fewer questions are needed: 
there is a $d$-detecting family of size $\Oh(|U|/\log |U|)$, for every fixed $d$~\cite{lindstrom1965combinatorial}. 
The logarithmic factor in the denominator will be crucial for deriving our lower bound.

Let us now sketch how $d$-detecting families are used in the proof of Theorem~\ref{th:listab}. 
Given an instance $\varphi$ of {\sc{3-SAT}} with $n$ variables and $\Oh(n)$ clauses, and a number $b\leq n/\log n$, we will construct an instance $G$ of \probLab[$(a$:$b)$] for some $a$.
This instance will have a positive answer if and only if $\varphi$ is satisfiable, and the constructed graph $G$ will have $\Oh(n/\log b)$ vertices.
It can be easily seen that this will yield the promised lower bound.

Partition the clause set $C$ of $\varphi$ into groups $C_1,C_2,\ldots,C_p$, each of size roughly $b$; thus $p=\Oh(n/b)$.
Similarly, partition the variable set $V$ of $\varphi$ into groups $V_1,\ldots,V_q$, each of size roughly $\log_2 b$; thus $q=\Oh(n/\log b)$.
In the output instance we create one vertex per each variable group---hence we have $\Oh(n/\log b)$ such vertices---and one block of vertices per each clause group, whose size will be determined in a moment.
Our construction ensures that the set of colors assigned to a vertex created for a variable group misses one color from some subset of $b$ colors.
The choice of the missing color corresponds to one of $2^{\log_2 b}=b$ 
possible boolean assignments to the variables of the group.

Take any vertex $u$ from a block of vertices created for some clause group $C_j$. We make it adjacent to vertices constructed for precisely those variable groups $V_i$, 
for which there is some variable in $V_i$ that occurs in some clause of $C_j$. 
This way, $u$ can only take a subset of the above missing colors corresponding to the chosen assignment on variables relevant to $C_j$.
By carefully selecting the list of $u$, and some additional technical gadgeteering, we can express a constraint of the following form:
the total number of satisfied literals in some subset of clauses of $C_j$ is exactly some number. Thus, we could verify that every clause of $C_j$ is satisfied by creating a block of $|C_j|$
vertices, each checking one clause. However, the whole graph output by the reduction would then have $\Oh(n)$ vertices, and we would not obtain any non-trivial lower bound. 
Instead, we create one vertex per each question in a $d$-detecting family on the universe $U=C_j$, which has size $\Oh(|C_j|/\log |C_j|)=\Oh(|C_j|/\log b)$.
Then, the total number of vertices in the constructed graph will be $\Oh(n/\log b)$, as intended.

\bigskip

Finally, we observe that from our main result one can infer a lower bound for the complexity of the \probMonoTest problem. Recall that in this problem
we are given an arithmetic circuit that evaluates a homogenous polynomial $P(x_1,x_2,\ldots,x_n)$ over some field $\mathbb{F}$; here, a polynomial is homogenous if
all its monomials have the same total degree $k$. The task is to verify whether $P$ has some monomial in which every variable has individual degree not larger than $r$, for a given parameter $r$.
Abasi et al.~\cite{abasi} gave a randomized algorithm solving this problem in time $\Ohstar(2^{\Oh(k\cdot \frac{\log r}{r})})$, where $k$ is the degree of the polynomial, assuming
that $\mathbb{F}=\field{p}$ for a prime $p\leq 2r^2+2r$. This algorithm was later derandomized by Gabizon et al.~\cite{GabizonLP15} within the same running time, but under the assumption
that the circuit is {\em{non-cancelling}}: it has only input, addition, and multiplication gates. Abasi et al.~\cite{abasi} and Gabizon et al.~\cite{GabizonLP15} gave a number of
applications of low-degree monomial detection to concrete problems. For instance, {\sc{$r$-Simple $k$-Path}}, the problem of finding a walk of length $k$ that visits every vertex at most $r$ times, can
be solved in time $\Ohstar(2^{\Oh(k\cdot \frac{\log r}{r})})$. However, for {\sc{$r$-Simple $k$-Path}}, as well as other problems that can be tackled using this technique, the best known
lower bounds under ETH exclude only algorithms with running time $\Ohstar(2^{o(\frac{k}{r})})$. Whether the $\log r$ factor in the exponent is necessary was left open by Abasi et al. and Gabizon et al.

We observe that the \probLab problem can be reduced to \probMonoTest over the field $\field{2}$ in such a way that an $\Ohstar(2^{k\cdot o(\frac{\log r}{r})})$-time algorithm for the latter
would imply a $2^{o(\log b)\cdot n}$-time algorithm for the former, which would contradict ETH. Thus, we show that the known algorithms for \probMonoTest  most probably cannot be sped up
in general; nevertheless, the question of lower bounds for specific applications remains open. However, going through \probLab to establish a lower bound for \probMonoTest is actually quite a detour,
because the latter problem has a much larger expressive power. Therefore, we also give a more straightforward reduction that starts from a convenient form of {\sc{Subset Sum}}; this reduction also
proves the lower bound for a wider range of $r$, expressed as a function of $k$.

\paragraph*{Outline.} In Section~\ref{sect:pre} we set up the notation as well as recall definitions and well-known facts.
We also discuss $d$-detecting families, the main combinatorial tool used in our reduction.
In Section~\ref{sect:listab} we prove the lower bound for the list version of the problem, i.e., Theorem~\ref{th:listab}.
In Section~\ref{sect:nolist} we give a reduction from the list version to the standard version, thereby proving Theorem~\ref{th:main}.
Section~\ref{sect:ldtesting} is devoted to deriving lower bounds for low-degree monomial testing.

\section{Preliminaries}\label{sect:pre}

\paragraph*{Notation.} We use standard graph notation, see e.g.~\cite{platypus,Diestel-book}.
All graphs we consider in this paper are simple and undirected.
For an integer $k$, we denote $[k]=\{0,\ldots,k-1\}$.
By $\uplus$ we denote the disjoint union, i.e., by $A \uplus B$ we mean $A\cup B$ with the indication that $A$ and $B$ are disjoint.
If $I$ and $J$ are instances of decision problems $P$ and $R$, respectively, then we say that $I$ and $J$ are {\em equivalent}
if either both $I$ and $J$ are YES-instances of respective problems, or both are NO-instances.

\paragraph*{Exponential-Time Hypothesis.}
The Exponential Time Hypothesis (ETH) of Impagliazzo et al.~\cite{eth} states that there exists a constant $c > 0$, such that there is no algorithm solving \probTSAT in time $\Ohstar(2^{cn})$.
During the recent years, ETH became the central conjecture used for proving tight bounds on the complexity of various problems. 
One of the most important results connected to ETH is the {\em{Sparsification Lemma}}~\cite{seth}, which essentially gives a reduction from an arbitrary instance of {\sc{$k$-SAT}} to an instance where
the number of clauses is linear in the number of variables.
The following well-known corollary can be derived by combining ETH with the Sparsification Lemma.

\begin{theorem}[see e.g.~Theorem 14.4 in \cite{platypus}]
\label{thm:eth-main}
Unless ETH fails, there is no algorithm for \probTSAT that runs in time $2^{o(n+m)}$, where $n,m$ denote the numbers of variables and clauses, respectively.
\end{theorem}

We need the following regularization result of Tovey~\cite{Tovey84}.
Following Tovey, by \probTFSAT we call the variant of \probTSAT where each clause of the input formula contains exactly $3$ different variables, and each variable occurs in at most $4$ clauses.

\begin{lemma}[\cite{Tovey84}]
	\label{lem_transformation}
	Given a \probTSAT formula $\varphi$ with $n$ variables and $m$ clauses one can transform
	it in polynomial time into an equivalent \probTFSAT instance $\varphi'$ with $\Oh(n+m)$ 
	variables and clauses.
\end{lemma}


\begin{corollary}
 \label{cor:eth-3,4-sat}
Unless ETH fails, there is no algorithm for \probTFSAT that runs in time $2^{o(n)}$, where $n$ denotes the number of variables of the input formula. 
\end{corollary}

\paragraph*{List and nonuniform list ($a$:$b$)-coloring}
For integers $a,b$ and a graph $G$ with a function $L\colon V(G) \to 2^{[a]}$ (assigning a list of colors to every vertex), an \emph{$L$-($a$:$b$)-coloring} of $G$ is an assignment of exactly $b$ colors from $L(v)$ to each vertex $v\in V(G)$, such that adjacent vertices get disjoint color sets. 
The \probLab problem asks, given $(G,L)$, whether an $L$-($a$:$b$)-coloring of $G$ exists.

As an intermediary step of our reduction, we will use the following generalization of list colorings where the number of demanded colors varies with every vertex.
For integers $a,b$, a graph $G$ with a function $L\colon V(G) \to 2^{[a]}$ and a \emph{demand function} $\beta\colon V(G) \to \{1,\dots,b\}$, an \emph{$L$-($a$:$\beta$)-coloring} of $G$ is an assignment of exactly $\beta(v)$ colors from $L(v)$ to each vertex $v \in V(G)$, such that adjacent vertices get disjoint color sets. \probNULab is then the problem in which given $(G,L,\beta)$ we ask if an $L$-($a$:$\beta$)-coloring of $G$ exists.

\paragraph*{$d$-detecting families.}
In our reductions the following notion plays a crucial role.

\begin{definition}\label{def:uniftest}
A \emph{$d$-detecting family} for a finite set $U$ is a family $\mathcal{F} \subseteq 2^U$ of subsets of $U$ such that for every two functions $f,g:U\to\{0,\dots,d - 1\}$, $f \neq g$, there is a set $S$ in the family such that $\sum_{x \in S} f(x) \neq \sum_{x \in S} g(x)$.
\end{definition}

A deterministic construction of sublinear, $d$-detecting families was given by Lindstr{\"o}m~\cite{lindstrom1965combinatorial}, together with a proof that even the constant factor 2 in the family size cannot be improved.

\begin{theorem}[\cite{lindstrom1965combinatorial}]\label{th:uniftest}
For every constant $d \in \mathbb{N}$ and finite set $U$, there is a $d$-detecting family $\mathcal{F}$ on $U$ of size $\frac{2|U|}{\log_d |U|} \cdot \left(1 + o(1)\right)$. Furthermore, $\mathcal{F}$ can be constructed in time polynomial in $|U|$.
\end{theorem}

Other constructions, generalizations, and discussion of similar results can be found in Grebinski and Kucherov~\cite{GrebinskiK00}, and in Bshouty~\cite{Bshouty09}.
Note that the expression $\sum_{x \in S} f(x)$ is just the product of $f$ as a vector in $[d]^{|U|}$ with the characteristic vector of $S$. 
Hence, instead of subset families, Lindstr{\"o}m speaks of \emph{detecting vectors}, while later works see them as \emph{detecting matrices}, 
that is, $(0,1)$-matrices with these vectors as rows (which define an injection on $[d]^{|U|}$ despite having few rows).
Similar definitions appear in the study of query complexity, e.g., as in the popular Mastermind game~\cite{Chvatal83}.

While known polynomial deterministic constructions of detecting families involve some number theory or fourier analysis, their existence can be argued with an elementary probabilistic argument.
Intuitively, a random subset $S\subseteq U$ will distinguish two distinct functions $f,g:U\to\{0,\dots,d - 1\}$ 
(meaning $\sum_{x \in S} f(x) \neq \sum_{x \in S} g(x)$) with probability at least $\frac{1}{2}$. 
This is because any $x$ where $f$ and $g$ disagree is taken or not taken into $S$ with probability $\frac{1}{2}$,
while sums over $S$ cannot agree in both cases simultaneously, as they differ by $f(x)$ and $g(x)$ respectively.
There are $d^{n} \cdot d^{n}$ function pairs to be distinguished.
In any subset of pairs, at least half are distinguished by a random set in expectation, thus at least one such set exists. 
Repeatedly finding such a set for undistinguished pairs, we get $|\log_{\frac{1}{2}} (d^{n} \cdot d^n)| = \Oh(n \log d)$ sets that distinguish all functions.
More strongly though, when two functions differ on more values, the probability of distinguishing them increases significantly. Hence we need fewer random sets to distinguish all pairs of distant functions. On the other hand, there are few function pairs that are close, so we need few random sets to distinguish them all as well.
This allows to show that in fact $\Oh(\frac{n}{\log_d n})$ random sets are enough to form a $d$-detecting family with positive probability~\cite{GrebinskiK00}.

\section{Hardness of \probLab}\label{sect:listab}

In this section we show our main technical contribution: an  ETH-based lower bound for \probLab. 
The key part is reducing an $n$-variable instance \probTSAT to an instance of \probNULab with only $\Oh(\frac{n}{\log b})$ vertices. 
Next, it is rather easy to reduce \probNULab to \probLab. We proceed with the first, key part.

\subsection{The nonuniform case}

We prove the following theorem through the remaining part of this section.

\begin{theorem}\label{th:non-unif-abconstruction}
For any instance $\phi$ of \probTFSAT with $n$ variables and any integer $2 \leq b\leq n/\log_2 n$, there is an equivalent instance $(G,\beta,L)$ of \probNULab[($a$:$2b$)] such that $a=\Oh(b^2\log b)$, $|V(G)|=\Oh(\frac{n}{\log b})$ and $G$ is 3-colorable.
Moreover, the instance $(G,\beta,L)$ and the 3-coloring of $G$ can be constructed in $\poly(n)$ time.
\looseness=-1
\end{theorem}

%

Consider an instance $\phi$ of \probTSAT where each variable appears in at most four clauses.
Let $V$ be the set of its variables and $C$ be the set of its clauses.
Note that $\frac{1}3|V| \le |C| \le \frac{4}3 |V|$. Let $a = 12 b^2 \cdot \lfloor\log_2 b\rfloor$.
We shall construct, for some integers $n_V = \Oh(|V|/\log b)$ and $n_C = \Oh(|C|/b)$:
\begin{itemize}
	\item a partition $V=V_1 \uplus \ldots \uplus V_{n_V}$ of variables into groups of size at most $\lfloor\log_2 b\rfloor$,
 	\item a partition $C=C_1 \uplus \ldots \uplus C_{n_C}$ of clauses into groups of size at most $b$,
 	\item a function $\sigma \colon \{1,\dots,n_V\} \to [12 \cdot b \cdot \lfloor\log_2 b\rfloor]$, 
\end{itemize}
such that the following condition holds:
\begin{equation}\tag{$\maltese$}\label{eq:diffvar}\begin{minipage}{0.85\textwidth}
	For any $j=1,\ldots,n_C$, the variables occurring in clauses of $C_j$ are all different and they all belong to pairwise different variable groups.
	Moreover, the indices of these groups are mapped to pairwise different values by $\sigma$.
\end{minipage}\end{equation}
In other words, any two literals of clauses in $C_j$ have different variables, and if they belong to $V_i$ and $V_{i'}$ respectively, then $\sigma(i)\neq \sigma(i')$.
\begin{lemma}
Partitions $V=V_1 \uplus \ldots \uplus V_{n_V}$, $C=C_1 \uplus \ldots \uplus C_{n_C}$ and a function $\sigma$ satisfying~\eqref{eq:diffvar} can be found in time $\Oh(n)$.
\end{lemma}
\begin{proof}
We first group variables, in a way such that the following holds: (P1) the variables occurring in any clause are different and belong to different variable groups. 
To this end, consider the graph $G_1$ with variables as vertices and edges between any two variables that occur in a common clause (i.e. the primal graph of $\phi$).
Since no clause contains repeated variables, $G_1$ has no loops.
Since every variable of $\phi$ occurs in at most four clauses, and since those clauses contain at most two other variables, the maximum degrees of $G_1$ is at most 8.
Hence $G_1$ can be greedily colored with 9 colors.
Then, we refine the partition given by colors to make every group have size at most $\lfloor\log_2 b\rfloor$, producing in total at most $n_V := \lceil|V|/\lfloor \log_2 b \rfloor\rceil +9$ groups $V_1,\dots,V_{n_V}$.
(P1) holds, because any two variables occurring in a common clause are adjacent in $G_1$, and thus get different colors, and thus are assigned to different groups.

Next, we group clauses in a way such that: (P2) the variables occurring in clauses of a group $C_j$ are all different and belong to different variable groups. 
For this, consider the graph $G_2$ with clauses as vertices, and with an edge between clauses if they contain two different variables from the same variable group. 
By (P1), $G_2$ has no loops.
Since every clause contains exactly 3 variables, each variable is in a group with at most $\lfloor\log_2 b\rfloor-1$ others, and every such variable occurs in at most 4 clauses, the maximum degree of $G_2$ is at most $12 (\lfloor\log_2 b\rfloor -1)$. 
We can therefore color $G_2$ greedily with $12 \lfloor\log_2 b\rfloor$ colors.
Similarly as before, we partition clauses into $n_C := \lceil |C|/b \rceil + 12\lfloor\log_2 b\rfloor$ 
monochromatic groups $C_1,\dots,C_{n_C}$ of size at most $b$ each. Then (P2) holds by construction of the coloring.

Finally, consider a graph $G_3$ with variable groups as vertices, and with an edge between two variable groups if they contain two different variables occurring in clauses from a common clause group.
More precisely, $V_i$ and $V_{i'}$ are adjacent if there are two different variables $x\in V_i$ and $x'\in V_{i'}$, and a clause group $C_j$ with clauses $c$ and $c'$ (possibly $c=c'$), such that $x$ occurs in $c$ and $x'$ occurs in $c'$.
By (P2), $G_3$ has no loops.
Since a variable has at most $\lfloor\log_2 b\rfloor -1$ other variables in its group, each of these variables occur in at most 4 clauses, each of these clauses has at most $b-1$ other clauses in its group, and each of these contains exactly 3 variables, the maximum degree of $G_3$ is at most $4 \cdot (\lfloor\log_2 b\rfloor -1) \cdot (b-1) \cdot 3$.
We can therefore color it greedily into $12 b \lfloor \log_2 b \rfloor$ colors. Let $\sigma$ be the resulting coloring.
By (P2) and the construction of this coloring, \eqref{eq:diffvar} holds.

The colorings can be found in linear time using standard techniques.
Note that we have $n_V = \lceil |V|/\lfloor \log_2 b \rfloor\rceil +9 = \Oh(|V|/\log b)$.
Moreover, since $b \leq  n / \log_2 n$, we get $\log_2 b \le \log_2 n \le \frac{n}b = \Theta({|C|}/b)$ and hence $n_C = \lceil|C|/b\rceil + 12\lfloor\log_2 b\rfloor = \Oh(|C|/b)$.
\end{proof}

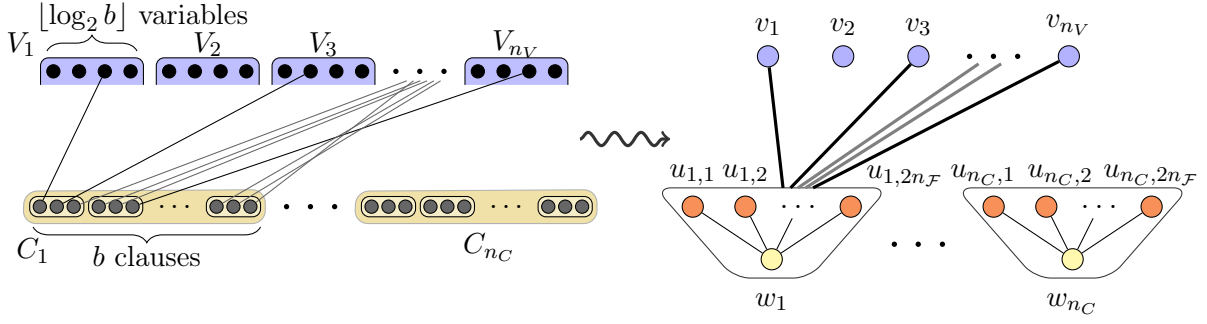
\begin{figure}[H]
	\centering
	\tikzset{ %
	V/.style={circle,draw=black,inner sep=0pt,minimum size=5pt,fill=black},
	L/.style={V, fill=gray!80!black},	
	v/.style={circle,draw=black,inner sep=0pt,minimum size=8pt,fill=blue!30},
	u/.style={circle,draw=black,inner sep=0pt,minimum size=8pt,fill=yellow!25!red!70},	
	w/.style={circle,draw=black,inner sep=0pt,minimum size=8pt,fill=yellow!40},		
	vg/.style={fill=blue!25},
	cla/.style={fill=yellow!60!brown!90!gray!30!white},
	clg/.style={draw=black!30,fill=yellow!60!brown!90!gray!50!white}
}
\begin{tikzpicture}
\begin{scope}[xscale=0.9,yscale=0.9]
\begin{scope}[xscale=0.95]
	\begin{scope}
		\node at (-0.5,0.4) { $V_1$ };
		\draw [decorate,decoration={brace,amplitude=6pt},xshift=0pt,yshift=0pt]
			(-0.1,0.3) -- (1.3,0.3)
			node [midway,yshift=13pt,xshift=19pt] {$\lfloor \log_2 b\rfloor$ variables};
		\draw[vg] (-0.2, -0.2) {[rounded corners = 3pt]-- (-0.2,0.2) -- (1.4,0.2)} -- (1.4,-0.2);
		\node[V] (var11) at (0.0,0) {};
		\node[V] (var12) at (0.4,0) {};
		\node[V] (var13) at (0.8,0) {};
		\node[V] (var14) at (1.2,0) {};	
	\end{scope}

	\begin{scope}[shift={(1.8,0)}]
		\node at (0.6,0.4) { $V_2$ };
		\draw[vg] (-0.2, -0.2) {[rounded corners = 3pt]-- (-0.2,0.2) -- (1.4,0.2)} -- (1.4,-0.2);
		\node[V] (var21) at (0.0,0) {};
		\node[V] (var22) at (0.4,0) {};
		\node[V] (var23) at (0.8,0) {};
		\node[V] (var24) at (1.2,0) {};	
	\end{scope}
	
	\begin{scope}[shift={(3.6,0)}]
		\node at (0.6,0.4) { $V_3$ };
		\draw[vg] (-0.2, -0.2) {[rounded corners = 3pt]-- (-0.2,0.2) -- (1.4,0.2)} -- (1.4,-0.2);
		\node[V] (var31) at (0.0,0) {};
		\node[V] (var32) at (0.4,0) {};
		\node[V] (var33) at (0.8,0) {};
		\node[V] (var34) at (1.2,0) {};	
	\end{scope}	

	\begin{scope}[shift={(5.1,0)}]
		\node at (0.6,0) {\huge $\ldots$};
	\end{scope}	

	\begin{scope}[shift={(6.6,0)}]
		\node at (0.6,0.4) { $V_{n_V}$ };	
		\draw[vg] (-0.2, -0.2) {[rounded corners = 3pt]-- (-0.2,0.2) -- (1.4,0.2)} -- (1.4,-0.2);
		\node[V] (varn1) at (0.0,0) {};
		\node[V] (varn2) at (0.4,0) {};
		\node[V] (varn3) at (0.8,0) {};
		\node[V] (varn4) at (1.2,0) {};	
	\end{scope}	
\end{scope}
	
	\begin{scope}[shift={(-0.2,-2)}]
		\node at (-0.1,-0.65) { $C_1$ };
		\draw [decorate,decoration={brace,amplitude=6pt,mirror},xshift=0pt,yshift=0pt]
		(-0.1,-0.3) -- (3.26,-0.3)
		node [midway,yshift=-11pt,xshift=0pt] {$b$ clauses};
		\draw[clg] (-0.23, 0.23) {[rounded corners = 5pt]-- (-0.25,-0.25) -- (3.33,-0.25) -- (3.33,0.25) -- cycle};
		\begin{scope}
			\draw[cla] (-0.15, -0.15) {[rounded corners = 3pt]-- (-0.15,0.15) -- (0.65,0.15) -- (0.65,-0.15) -- cycle};
			\node[L] (l11) at (0.0,0) {};
			\node[L] (l12) at (0.25,0) {};
			\node[L] (l13) at (0.5,0) {};
		\end{scope}
		\begin{scope}[shift={(0.87,0)}]
			\draw[cla] (-0.15, -0.15) {[rounded corners = 3pt]-- (-0.15,0.15) -- (0.65,0.15) -- (0.65,-0.15) -- cycle};
			\node[L] (l21) at (0.0,0) {};
			\node[L] (l22) at (0.25,0) {};
			\node[L] (l23) at (0.5,0) {};
		\end{scope}
		\begin{scope}[shift={(0.87*2,0)}]
			\node at (0.25,0) {$\ldots$};
		\end{scope}
		\begin{scope}[shift={(0.87*3,0)}]
			\draw[cla] (-0.15, -0.15) {[rounded corners = 3pt]-- (-0.15,0.15) -- (0.65,0.15) -- (0.65,-0.15) -- cycle};
			\node[L] (l31) at (0.0,0) {};
			\node[L] (l32) at (0.25,0) {};
			\node[L] (l33) at (0.5,0) {};
		\end{scope}
	\end{scope}
	\begin{scope}[shift={(-0.55,-2)}]
		\node at (0.87*5,0) {\huge$\ldots$};
	\end{scope}
	\begin{scope}[shift={(4.7,-2)}]
		\node at (0.87*2,-0.6) { $C_{n_C}$ };
		\draw[clg] (-0.23, 0.23) {[rounded corners = 5pt]-- (-0.25,-0.25) -- (3.33,-0.25) -- (3.33,0.25) -- cycle};
		\begin{scope}
			\draw[cla] (-0.15, -0.15) {[rounded corners = 3pt]-- (-0.15,0.15) -- (0.65,0.15) -- (0.65,-0.15) -- cycle};
			\node[L] (l1) at (0.0,0) {};
			\node[L] (l2) at (0.25,0) {};
			\node[L] (l3) at (0.5,0) {};
		\end{scope}
		\begin{scope}[shift={(0.87,0)}]
			\draw[cla] (-0.15, -0.15) {[rounded corners = 3pt]-- (-0.15,0.15) -- (0.65,0.15) -- (0.65,-0.15) -- cycle};
			\node[L] (l1) at (0.0,0) {};
			\node[L] (l2) at (0.25,0) {};
			\node[L] (l3) at (0.5,0) {};
		\end{scope}
		\begin{scope}[shift={(0.87*2,0)}]
			\node at (0.25,0) {$\ldots$};
		\end{scope}
		\begin{scope}[shift={(0.87*3,0)}]
			\draw[cla] (-0.15, -0.15) {[rounded corners = 3pt]-- (-0.15,0.15) -- (0.65,0.15) -- (0.65,-0.15) -- cycle};
			\node[L] (l1) at (0.0,0) {};
			\node[L] (l2) at (0.25,0) {};
			\node[L] (l3) at (0.5,0) {};
		\end{scope}
	\end{scope}
	
	\draw (l11) -- (var13);
	\draw (l12) -- (var32);
	\draw[black!60] (l13) -- (5.22,-0.14);
	\draw[black!60] (l21) -- (5.42,-0.14);
	\draw[black!60] (l22) -- (5.52, -0.14);
	\draw (l23) -- (varn3);	
	\draw[black!60] (l31) -- (5.62,-0.14);
	\draw[black!60] (l32) -- (5.72,-0.14);	
	\draw[black!60] (l33) -- (5.32,-0.14);
\end{scope}

\draw [very thick,black!80,->,decorate,decoration=snake] (7.0,-0.9) -- (8.2,-0.9);

\begin{scope}[shift={(9.5,0.2)}]
	\node[v,label=above:$v_1$] (v1) at (0.0,0) {};
	\node[v,label=above:$v_2$] (v2) at (1.0,0) {};
	\node[v,label=above:$v_3$]  (v3) at (2.0,0) {};
	\node (vd) at (3.0,0) {\huge$\ldots$};
	\node[v,label=above:$v_{n_V}$] (vn) at (4.0,0) {};		
	
	\begin{scope}[shift={(-1,-2)}]
		\draw[rounded corners=6pt,black!80] (-0.5,0.25) -- (2.6,0.25) -- (1.55,-0.95) -- (0.55,-0.95) -- cycle;
		\node[u,label=above:$u_{1,1}$] (u1) at (0,0) {};
		\node[u,label=above:$u_{1,2}$] (u2) at (0.7,0) {};
		\node (ud)  at (1.4, 0) {\large$\ldots$};
		\node[u,label=60:$u_{1,2n_\Ff}$] (un) at (2.1,0) {};	
		\node[w,label=below:$w_1\strut$] (w0) at (1.05, -0.7) {};
		\draw (w0) -- (u1);
		\draw (w0) -- (u2);
		\draw (w0) -- (ud);
		\draw (w0) -- (un);
		\draw[very thick] (1.2,0.25) -- (v1);
		\draw[very thick] (1.3,0.25) -- (v3);
		\draw[very thick,black!50] (1.4,0.25) -- (3.8,1.9);	
		\draw[very thick,black!50] (1.5,0.25) -- (4.1,1.9);
		\draw[very thick] (1.6,0.25) -- (vn);
	\end{scope}
	\node at (2,-2.5) {\huge$\ldots$};
	\begin{scope}[shift={(3,-2)}]
		\draw[rounded corners=6pt,black!80] (-0.5,0.25) -- (2.6,0.25) -- (1.55,-0.95) -- (0.55,-0.95) -- cycle;
		\node[u,label={[xshift=-4pt]above:$u_{n_C,1}$}] (u1) at (0,0) {};
		\node[u,label={[xshift=-12pt]75:$u_{n_C,2}$}] (u2) at (0.7,0) {};
		\node (ud)  at (1.4, 0) {\large$\ldots$};
		\node[u,label={[xshift=0pt]90:$u_{n_C,2n_\Ff}$}] (un) at (2.1,0) {};	
		\node[w,label=below:$w_{n_C}\strut$] (w0) at (1.05, -0.7) {};
		\draw (w0) -- (u1);
		\draw (w0) -- (u2);
		\draw (w0) -- (ud);
		\draw (w0) -- (un);
	\end{scope}	
\end{scope}
\end{tikzpicture}
	\vspace*{-23pt}
	\caption{\emph{(left)} The groups of variables and clauses of the formula; literals in $C_1$ are joined with their variables. Since no variable of $V_2$ occurs in $C_1$, we have $2 \not \in I_1$ -- this may allow us to make $\sigma(2)$ the same number as $\sigma(3)$, say, reducing the total number $a$ of colors needed.
	\emph{(right)} The constructed graph; thick lines represent edges to all vertices corresponding to $C_1$.}
	\label{fig:reduction}
\end{figure}

For every $1 \leq i \leq n_V$, the set $V_i$ of variables admits $2^{|V_i|}\leq b$ different assignments. 
We will therefore say that each assignment on $V_i$ is \emph{given} by an integer $x \in [b]$, for example by interpreting the first $|V_i|$ bits of the binary representation of $x$ as truth values for variables in $V_i$.
Note that when $|V_i| < \log_2 b$, different integers from $[b]$ may give the same assignment on $V_i$.\looseness=-1
%
%
\smallskip

For $1\leq j \leq n_C$, let $I_j \subseteq \{1, \ldots, n_V\}$ be the set of indices of variable groups that contain some variable occurring in the clauses of $C_j$.
Since every clause contains exactly three literals, property~\eqref{eq:diffvar} means that $|I_j|=3|C_j|$ and that $\sigma$ is injective over each $I_j$.
See Fig.~\ref{fig:reduction}.
\smallskip

For $1\leq j \leq n_C$, let $\{C_{j,1},\dots,C_{j,n_\Ff}\}$ be a 4-detecting family of subsets of $C_j$, for some $n_\Ff = \Oh(\frac{b}{\log b})$ (we can assume $n_\Ff$ does not depend on $j$ by adding arbitrary sets when $|C_j|<b$).
For every $1 \leq k \leq n_\Ff$, let $C_{j,n_\Ff+k}=C_j\setminus C_{j,k}$.

We are now ready to build the graph $G$, the demand function $\beta:V(G)\rightarrow\{1,\ldots,2b\}$, and the list assignment $L$ as follows.

\begin{enumerate}[(1)]
\item For $1 \leq i \leq n_V$, create a vertex $v_i$ with $\beta(v_i)=b-1$ and $L(v_i)=\{b\cdot \sigma(i) +x \mid x \in [b]\}$.
\item
For $1 \leq j \leq n_C$ and $1 \leq k \leq 2n_\Ff$, create a vertex $u_{j,k}$ adjacent to each $v_i$ for $i \in I_j$.\\
Let $\beta(u_{j,k})=|C_{j,k}|$ and
\begin{align*}
	L(u_{j,k})=\{b\cdot &\sigma(i)+x\ \mid\ 1\leq i \leq n_V, x\in [2^{|V_i|}]\text{ such that}\\
	&x\text{ gives an assignment of }V_i\text{ that satisfies some clause of }C_{j,k}\}.
\end{align*}
\item\vspace*{-10pt} For $1 \leq j \leq n_C$, create a vertex $w_j$, adjacent to each $v_i$ for $i \in I_j$ and to each $u_{j,k}$ ($1 \leq k \leq 2n_\Ff$). Let $\beta(w_j)=2|C_j|$ and $L(w_j) = \bigcup_{i\in I_j} \{b \cdot \sigma(i) + x \mid x \in [b]\}$.
\end{enumerate}

Before giving a detailed proof of the correctness, let us describe the reduction in intuitive terms.
%
Note that vertices of type $v_i$ get all but one color from their list; this missing color, say $b\cdot \sigma(i) +x_i$, for some $x_i\in[b]$, defines an assignment on $V_i$. 
For every $j=1,\ldots,n_C$ the goal of the gadget consisting of $w_j$ and vertices $u_{j,k}$ is to express the constraint that every clause in $C_j$ has a literal satisfied by this assignment. Since $w_j,u_{j,k}$ are adjacent to all vertices in $\{v_i \mid i \in I_j\}$, they may only use the missing colors (of the form $b\cdot \sigma(i) +x_i$, where $i \in I_j$).
Since $|I_j|=3|C_j|$, there are $3|C_j|$ such colors and $2|C_j|$ of them go to $w_j$.
This leaves exactly $|C_j|$ colors for vertices of type $u_{j,k}$, corresponding to a choice of $|C_j|$ satisfied literals from the $3|C_j|$ literals in clauses of $C_j$.
The lists and demands for vertices $u_{j,k}$ guarantee that exactly $|C_{j,k}|$ chosen satisfied literals occur in clauses of $C_{j,k}$.
The properties of 4-detecting families will ensure that every clause has exactly one chosen, satisfied literal, and hence at least one satisfied literal.
We proceed with formal proofs.

\begin{lemma}\label{lem:PokGok}
If $\phi$ is satisfiable then $G$ is $L$-($a$:$\beta$)-colorable.
\end{lemma}

\begin{proof}
	Consider a satisfying assignment $\eta$ for $\phi$.
	For $1\leq i \leq n_V$, let $x_i\in [2^{|V_i|}]$ be an integer giving the same assignment on $V_i$ as $\eta$.
	For every clause $c$ of $\phi$, choose one literal satisfied by $\eta$ in it, and let $i_c$ be index of the group $V_{i_c}$ containing the literal's variable.
	Let $\alpha: V(G) \to \binom{[a]}{\le 2b}$ be the $L$-($a$:$\beta$)-coloring of $G$ defined as follows,  for $1\leq i \leq n_V$, $1\leq j \leq n_C$, $1\leq k \leq 2n_\Ff$:
	\begin{itemize}
		\item $\alpha(v_i) = L(v_i) \setminus \{b \cdot \sigma(i) + x_i\}$
		\item $\alpha(u_{j,k}) = \{ b \cdot \sigma(i_c) + x_{i_c} \mid c \in C_{j,k}\}$
		\item $\alpha(w_j)=\{b \cdot \sigma(i) + x_i  \mid i \in I_j \setminus \{i_c \mid c\in C_j\} \}$.
	\end{itemize}
	
Let us first check that every vertex $v$ gets colors from its list $L(v)$ only.
This is immediate for vertices $v_i$ and $w_j$, while for $u_{j,k}$ it follows from the fact that $x_{i_c}$ gives a partial assignment to $V_i$ that satisfies some clause of $C_{j,k}$.

Now let us check that for every vertex $v$, the coloring $\alpha$ assigns exactly $\beta(v)$ colors to $v$.
For $\alpha(v_i)$ this follows from the fact that $|L(v_i)|=b$ and $0 \leq x_i < 2^{|V_i|} \leq b$.
Since by property~\eqref{eq:diffvar}, $\sigma$ is injective on $I_j$, and thus on $\{i_c \mid c \in C_{j,k}\} \subseteq I_j$, we have $|\alpha(u_{j,k})|=|C_{j,k}|=b(u_{j,k})$.
Similarly, since $\sigma$ is injective on $I_j$ and $|I_j \setminus \{i_c \mid c\in C_j\}| = 3|C_j|-|C_j|=2|C_j|$, we get  $|\alpha(w_{j})|=2|C_{j}|=\beta(w_{j})$.
	
It remains to argue that the sets assigned to any two adjacent vertices are disjoint. There are three types of edges in the graph, namely $v_i u_{j,k}$, $v_i w_j$, and $w_j u_{j,k}$.
The disjointness of $\alpha(w_j)$ and $\alpha(u_{j,k})$ is immediate from the definition of $\alpha$, since $C_{j,k}\subseteq C_j$.
Fix $j=1,\ldots,n_C$.
Since $\sigma$ is injective on $I_j$, for any two different $i,i'\in I_j$, we have $b \cdot \sigma(i) + x_i \not\in L(v_{i'})$. Hence,
\[
 \bigcup_{i\in I_j} \alpha(v_i) = \{b \cdot \sigma(i) + x \mid i\in I_j\text{ and }x\in[b]\} \setminus \{b \cdot \sigma(i) + x_i \mid i\in I_j\}.
\]
Since $\alpha(u_{j,k}), \alpha(w_j) \subseteq \{b \cdot \sigma(i) + x_i \mid i\in I_j\}$, it follows that edges of types $v_i u_{j,k}$ and $v_i w_j$ received disjoint sets of colors on their endpoints, concluding the proof.
\end{proof}

\begin{lemma}\label{lem:GokPok}
If $G$ is $L$-($a$:$\beta$)-colorable then $\phi$ is satisfiable.
\end{lemma}
\begin{proof}
Assume that $G$ is $L$-($a$:$\beta$)-colorable, and let $\alpha$ be the corresponding coloring. 

For $1\leq i\leq n_V$, we have $|L(v_i)| = b$ and $|\alpha(v_i)|=b-1$, so $v_i$ misses exactly one color from its list. Let $b \cdot \sigma(i) + x_i$, for some $x_i \in [b]$, be the missing color.
We want to argue that the assignment $x$ for $\phi$ given by $x_i$ on each $V_i$ satisfies $\phi$.

Consider any clause group $C_j$, for $1 \leq j \leq n_C$.
Every vertex in $\{w_j\}\cup \{u_{j,k}\mid 1 \leq k \leq 2n_{\Ff}\}$ contains $\{v_i\mid i\in I_j\}$ in its neighborhood.
Therefore, the sets $\alpha(u_{j,k})$ and $\alpha(w_j)$ are disjoint from $\bigcup_{i\in I_j} \alpha(v_i)$.
Since $L(u_{j,k}), L(w_j) \subseteq \{b \cdot \sigma(i) + x' \mid i \in I_j, x' \in [b]\}$, we get that $\alpha(u_{j,k})$ and $\alpha(w_j)$ are contained in the set of missing colors $\{b \cdot \sigma(i) + x_i \mid i\in I_j\}$ (corresponding to the chosen assignment).
By property~\eqref{eq:diffvar}, this set has exactly $|I_j|=3|C_j|$ different colors.
Of these, exactly $2|C_j|$ are contained in $\alpha(w_j)$.
Let the remaining $|C_j|$ colors be $\{b \cdot \sigma(i) + x_i \mid i \in J_j\}$, for some subset $J_j\subseteq I_j$ of $|C_j|$ indices.

Since $\alpha(u_{j,k})$ is disjoint from $\alpha(w_j)$, we have $\alpha(u_{j,k})\subseteq \{b \cdot \sigma(i) + x_i \mid i \in J_j\}$ 
for all $k$.
By definition of $I_j$, for every $i \in J_j \subseteq I_j$ there is a variable in $V_i$ that appears in some clause of $C_j$. 
By property~\eqref{eq:diffvar}, it can only occur in one such clause, so let $l_i$ be the literal in the clause of $C_j$ where it appears.
For every color $b\cdot \sigma(i) + x_i \in \alpha(u_{j,k})$, by definition of the lists for $u_{j,k}$ we know that $x_i$ gives 
a partial assignment to $V_i$ that satisfies some clause of $C_{j,k}$. This means $x_i$ makes the literal $l_i$ true and $l_i$ occurs in a clause of $C_{j,k}$.
Therefore, for each $k$, at least $|\alpha(u_{j,k})| = |C_{j,k}|$ literals from the set $\{l_i \mid i \in J_j\}$ occur in clauses of $C_{j,k}$ and are made true by the assignment $x$.

Let $f:C_j \to \{0,1,2,3\}$ be the function assigning to each clause $c\in C_j$ the number of literals of $c$ in $\{l_i \mid i \in J_j\}$.
By the above, $\sum_{c \in C_{j,k}} f(c) \geq |C_{j,k}|$ for $1\leq k \leq 2n_\Ff$.
Since each literal in $\{l_i \mid i \in J_j\}$ belongs to some clause of $C_j$, we have $\sum_{c\in C_j} f(c) = |J_j| = |C_j|$.
Then,
\[\sum_{c\in C_{j,k}}f(c) = \sum_{c\in C_{j}}f(c) - \sum_{c\in C_{j,n_\Ff+k}}f(c) \le |C_j| - |C_{j,n_\Ff+k}| = |C_{j,k}|.\]
Hence $\sum_{c \in C_{j,k}} f(c) = |C_{j,k}|$ for $1\leq k \leq 2n_\Ff$. Let $g:C_j\rightarrow\{0,1,2,3\}$ be the constant function $g \equiv 1$.
Note that 
\[\sum_{c\in C_{j,k}}g(c) = |C_{j,k}| = \sum_{c\in C_{j,k}}f(c).\]
Since $\{C_{j,1},\ldots,C_{j,n_\Ff}\}$ is a 4-detecting family, this implies that $f \equiv 1$.
Thus, for every clause $c$ of $C_j$ we have $f(c)=1$, meaning that there is a literal from the set $\{l_i \mid i \in J_j\}$ in this clause.
All these literals are made positive by the assignment $\eta$, therefore all clauses of $C_j$ are satisfied.
Since $j=1,\ldots,n_C$ was arbitrary, this concludes the proof that $\eta$ is a satisfying assignment for $\phi$.
\end{proof}

The construction can clearly be made in polynomial time and the total number of vertices is $n_V + n_C \cdot \Oh(\frac{b}{\log b}) + n_C = \Oh(\frac{n}{\log b})$.
Moreover, we get a proper 3-coloring of $G$, by coloring vertices of the type $v_i$ by color 1, vertices of the type $u_{j,k}$ by color 2, and vertices of the type $w_{j}$ by color 3.
By Lemmas~\ref{lem:PokGok} and~\ref{lem:GokPok}, this concludes the proof of Theorem~\ref{th:non-unif-abconstruction}.

\subsection{The uniform case}

In this section we reduce the nonuniform case to the uniform one, and state the resulting lower bound on the complexity of \probLab.

\begin{lemma}\label{lem:abconstruction}
For any instance $I=(G,\beta,L)$ of \probNULab[($a$:$b$)] where the graph $G$ is $t$-colorable, there is an equivalent instance $(G,L')$ of \probLab[($(a+tb)$:$b$)].
Moreover, given a $t$-coloring of $G$ the instance $(G,L')$ can be constructed in time polynomial in $|I|+b$.
\looseness=-1
\end{lemma}

\begin{proof}
Let $c\colon V(G)\rightarrow [t]$ be a $t$-coloring of $G$. For every vertex $v$, define a set of filling colors $F(v)=\{a + c(v)b + i\ \colon\ i=0,\ldots,b-|\beta(v)|-1\}$ and put
\[L'(v) = L(v) \cup F(v).\]

Let $\alpha\colon V(G)\rightarrow 2^{[a]}$ be an $L$-($a$:$\beta$)-coloring of $G$. 
We define a coloring $\alpha'\colon V(G)\rightarrow 2^{[a+tb]}$ by setting $\alpha'(v)=\alpha(v) \cup F(v)$ for every vertex $v\in V(G)$. 
Observe that $\alpha'(v)\subseteq L'(v)$ and $|\alpha'(v)|=|\alpha(v)| + (b-|\beta(v)|) = b$.
Since $\alpha$ was a proper $L$-($a$:$\beta$)-coloring, adjacent vertices can only share the filling colors. 
However, the lists of adjacent vertices have disjoint subsets of filling colors, since these vertices are colored differently by $c$. 
It follows that $\alpha'$ is an $L'$-($a$:$b$)-coloring of $G$.

Conversely, let $\alpha'\colon V(G)\rightarrow 2^{[a+tb]}$ be an $L'$-($a$:$b$)-coloring of $G$. 
For every vertex $v$, we have $|\alpha'(v)\cap[a]| = b - |\alpha'(v) \cap F(v)| \ge b - (b-|\beta(v)|) = |\beta(v)|$.
Define $\alpha(v)$ to be any cardinality $\beta(v)$ subset of $\alpha'(v)\cap[a]$.
It is immediate to check that $\alpha$ is an $L$-($a$:$\beta$)-coloring of $G$.
\end{proof}

We are now ready to prove one of our main results.

\listabthm*

\begin{proof}
Let $b(n)$ be a function as in the statement.
We can assume w.l.o.g.\ that $2\le b(n)\le n/\log_2 n$, for otherwise we can replace $b(n)$ with a function $b'(n)=2+\lfloor b(n)/c\rfloor$ in the reasoning below, where $c$ is a big enough constant; note that $b'(n)=\Theta(b(n))$.
Fix a function $g(b)=o(\log b)$ and assume there is an algorithm $\mathcal{A}$ for \probLab that runs in time~$2^{g(b) \cdot n}$, whenever $b=\Theta(b(n))$.
Consider an instance of \probTFSAT with $n$ variables. 
Let $b=b(n)$.
By Theorem~\ref{th:non-unif-abconstruction} in $\poly(n)$ time we get an equivalent instance $(G,\beta,L)$ of \probNULab[($a$:$(2b)$)] such that $a=\Theta(b^2 \log b)$, $|V(G)|=\Oh(\frac{n}{\log b})$, and a 3-coloring of $G$.
Next, by Lemma~\ref{lem:abconstruction} in $\poly(n)$ time we get an equivalent instance $(G,L')$ of \probLab[($(a+6b)$:$(2b)$)].
Finally, we solve the instance $(G,L')$ using algorithm $\mathcal{A}$.
Since $b(n)=\omega(1)$, we have $g(b(n))=o(\log(b(n)))$, and $\mathcal{A}$ runs in time $2^{o(\log b(n))\cdot |V(G)|}$.
Thus, we have solved the instance $\phi$ of \probTFSAT in time $2^{o(\log b(n))\cdot |V(G)|}=2^{o(\log b(n))\cdot \frac{n}{\log b(n)}}=2^{o(n)}$.
By Corollary~\ref{cor:eth-3,4-sat}, this contradicts ETH.
\end{proof}

\section{From \probLab to \probABCol}\label{sect:nolist}

In this section we reduce \probLab to \probABCol. This is done by adding a Kneser graph, and replacing the lists by edges to appropriate vertices of the Kneser graph.
We will need the following well-known property of Kneser graphs (see e.g., Theorem 7.9.1 in the textbook~\cite{godsil-agt}).

\begin{theorem}
\label{thm:homo}
 If $p>2q$ then every homomorphism from $KG_{p,q}$ to $KG_{p,q}$ is an automorphism.
\end{theorem}

We proceed with the reduction.

\begin{lemma}
\label{lem:remove-lists}
	For any given instance of \probLab with $n$ vertices,
	there exists an equivalent instance of \probABCol[($(a+b)$:$b$)]
	with $n + \binom{a + b}{b}$ vertices.
	Moreover, it can be computed in $\poly(n, \binom{a + b}{b})$-time.
\end{lemma}

\begin{proof}
	Let $(G, L)$ be an instance of \probLab
	where $G$ is a graph and $L\colon V(G)\to 2^{[a]}$ describes the lists
	of allowed colors.
	Define a graph $K$ with
	$V(K) = \binom{[a + b]}{b}$ and
	$$E(K) = \{\, XY\ \colon\ X, Y \in V(K) \textrm{ and } X \cap Y = \emptyset\, \}.$$ That is, $K$ is isomorphic to the Kneser graph $KG_{a+b,b}$.
	Then let
	$V' = V(G) \uplus V(K)$ and
	\[
		E' = E(G) \uplus E(K) \uplus \{\, vX\ \colon\ v \in V(G) \textrm{ and } X \in V(K) \textrm{ and } L(v) \cap X = \emptyset\, \}.
	\]
	The graph $G' = (V', E')$ has $n + \binom{a + b}{b}$ vertices, and the construction
	can be done in time polynomial in $n+\binom{a + b}{b}$.
	Let $G'$ be our output instance of \probABCol[($(a + b)$:$b$)].
	We will show that it is equivalent to the instance $(G, L)$ of \probLab.
	
	Let us assume that $\alpha:V(G)\to \binom{[a]}{b}$ is an $L$-$(a$:$b)$-coloring of $G$.
	Consider $\alpha': V(G') \to \binom{[a + b]}{b}$ such that
	\[
		\alpha'(v) = \begin{cases}
		\alpha(v) & \mbox{ for } v \in V(G)\\
		v & \mbox{ for } v \in V(K) = \binom{[a + b]}{b}.\\
		\end{cases}
	\]
	We claim that $\alpha'$ is an ($(a+b)$:$b$)-coloring of $G'$.
	Indeed, for every edge $uv \in E(G)$ we have $\alpha'(u) \cap \alpha'(v) = \alpha(u) \cap \alpha(v) = \emptyset$ because $\alpha$ is an $L$-$(a$:$b)$-coloring of $G$.
	For every edge $XY \in E(K)$ we have $\alpha'(X) \cap \alpha'(Y) = X \cap Y = \emptyset$.
	For every edge $vX \in E(V(G), V(K))$ we have $\alpha'(v) \cap \alpha'(X) = \alpha(v) \cap X \subseteq L(v) \cap X = \emptyset$.
	
	Now, let us assume that $\alpha': V(G') \to \binom{[a + b]}{b}$ is an ($(a + b)$:$b$)-coloring of $G'$.
	Recall that $\alpha'$ is a homomorphism of $G'$ to $KG_{a+b,b}$.
	Denote $\phi = \alpha'|_{V(K)}$.
	By Theorem~\ref{thm:homo}, $\phi$ is an automorphism of $K$. Define $\alpha''=\phi^{-1} \circ \alpha'$.
	Then $\alpha''$ is an ($(a + b)$:$b$)-coloring of $G'$ with the property that $\alpha''(X)=X$ for every $X\in V(K)$.
	We claim that $\alpha''|_{V(G)}$ is an $L$-$(a$:$b)$-coloring of $G$.
	Since $\alpha''$ is a ($(a + b)$:$b$)-coloring of $G'$, it suffices to show that $\alpha''(v)\subseteq L(v)$ for every vertex $v\in V(G)$.
	Pick a color $\gamma\not\in L(v)$.
	Let $X_\gamma$ be the $b$-element set consisting of $\gamma$ and arbitrary $b-1$ elements from $[a + b] \setminus ([a] \cup \{\gamma\})$.
	Then $L(v) \cap X_\gamma = \emptyset$ and hence $vX_\gamma \in E(G')$.
	It follows that $X_\gamma\cap \alpha''(v)=\alpha''(X_\gamma)\cap \alpha''(v)= \emptyset$, and in particular $\gamma \not\in \alpha''(v)$.
	Thus, $\alpha''(v)\subseteq L(v)$ as required.
\end{proof}

We now prove our main result.

\mainthm*

\begin{proof}
Fix a computable function $f(b)$, a function $g(b)=o(\log b)$ and assume there is an algorithm $\mathcal{A}$ for \probABCol that runs in time~$f(b)\cdot 2^{g(b) \cdot n}$ for a given $n$-vertex graph, whenever $a=\Theta(b^2 \log b)$.
Without loss of generality we can replace $f(b)$ by any non-decreasing function $f'(n)$ such that $f'(n)\geq f(n)$ and $f'(n)>n$.
Intuitively, we now define an unbounded function $b(N)$ which should be at least 2, at most the inverse of $f$, and small enough so that $2^{\Oh(b\log b)} \leq \frac{N}{\log b}$.
The following function is $\omega(1)$ and a standard argument shows how to compute it in $\poly(N)$ time (see Lemmas~3.2 and 3.4 in~\cite{CaiCDF95}).
$$b(N) = \min\left(\max\{b : f(b) \leq N\}\ ,\ \max\{b : b\log b \leq \log N / \log \log N\}\right) + 2.$$

Consider an instance of \probTFSAT with $N$ variables. 
Let $b=b(N)$.
By Theorem~\ref{th:non-unif-abconstruction} in $\poly(N)$ time we get an equivalent instance $(G,\beta,L)$ of \probNULab[($a$:$(2b)$)] such that $a=\Theta(b^2 \log b)$, $|V(G)|=\Oh(\frac{N}{\log b})$, and a 3-coloring of $G$.
Next, by Lemma~\ref{lem:abconstruction} in $\poly(N)$ time we get an equivalent instance $(G,L')$ of \probLab[($(a+6b)$:$(2b)$)].
Then, by Lemma~\ref{lem:remove-lists}, in time $\poly(N,\binom{a+8b}{2b})=\poly(N)$ we get an equivalent instance $G'$ of \probABCol[($(a+8b)$:$(2b)$)] such that $|V(G')|=|V(G)|+\binom{a+8b}{2b}$.
Observe that since $a=\Theta(b^2 \log b)$ and $b\log b \leq \log N / \log \log N$,
$${\textstyle\binom{a+8b}{2b}}\leq (a+8b)^{2b} = 2^{\Oh(b \log b)} = 2^{\Oh(\log N / \log \log N)} = N^{o(1)} = o({N}/{\log b(N)})$$
Hence $|V(G')| = \Oh(\frac{N}{\log b})$.
Finally, we solve the instance $G'$ using algorithm $\mathcal{A}$.
Since $b(N)=\omega(1)$, we have $g(b(N))=o(\log(b(N)))$.
Therefore, $\mathcal{A}$ runs in time
	$$f(b) \cdot 2^{o(\log b(N))\cdot |V(G')|} \leq N \cdot 2^{o(\log b(N)) \cdot \Oh(N / \log b(N))} = 2^{o(N)}$$
solving the instance $\phi$ of \probTFSAT in time $2^{o(N)}$.
By Corollary~\ref{cor:eth-3,4-sat}, this contradicts ETH.

\end{proof}

\maincor*
\begin{proof}
Fix a computable function $f(h)$ and assume there is an algorithm $\mathcal{A}$ for \probHom that runs in time~$f(h)\cdot 2^{o(\log\log h) \cdot n}$ for a given $n$-vertex graph, whenever  $H$ is a Kneser graph $K_{a,b}$ with $a=\Theta(b^2 \log b)$.
Consider an instance of \probABCol with $n$ vertices and $a=\Theta(b^2 \log b)$. 
This is an instance of \probHom with $h=\binom{a}{b}\leq a^b = 2^{\Oh(b \log b)}$, hence $\mathcal{A}$ solves it in
$$ f(h)\cdot 2^{o(\log \log h) \cdot n} = f(2^{\Oh(b \log b)}) \cdot 2^{o(\log (b\log b)) \cdot n} \leq f'(b) \cdot 2^{o(\log b) \cdot n}$$ 
for some computable function $f'(b) \geq f(2^{\Theta(b \log b)})$, which contradicts Theorem~\ref{th:main}.
\end{proof}

\section{Low-degree testing}\label{sect:ldtesting}


In this section we derive lower bounds for \probMonoTest. In this problem, we are given an arithmetic circuit $C$ over some field $\mathbb{F}$; such a circuit may contain
input, constant, addition, negation, multiplication, and inversion gates. One gate is designated to be the output gate, and it computes some polynomial $P$ of the variables $x_1,x_2,\ldots,x_n$
that appear in the input gates. We assume that $P$ is a homogenous polynomial of degree $k$, i.e., all its monomials have total degree $k$. The task is to verify whether $P$ contains an $r$-monomial, i.e., a monomial in which every variable has its individual degree bounded by $r$, for a given parameter $r\leq k$. Abasi et al.~\cite{abasi} gave a very fast randomized algorithm for \probMonoTest.

\begin{theorem}[Abasi et al.~\cite{abasi}]
\label{thm:abasi}
Fix integers $r, k$ with $2 \le r \le k$. 
Let $p\le 2r^2+2r$ be a prime, and let $g\in\field{p}[x_1,\ldots,x_n]$ be a homogenous polynomial of degree $k$, computable by a circuit $C$.
Then, there is a randomized algorithm running in time $\Oh(r^{2k/r}|C|(rn)^{\Oh(1)})$ which
\begin{itemize}
 \item with probability at least $1/2$ answers YES when $g$ contains an $r$-monomial,
 \item always answers NO when $g$ contains no $r$-monomial.
\end{itemize}
\end{theorem}

This result was later derandomized by Gabizon et al.~\cite{GabizonLP15} under the assumption that the circuit is {\em{non-cancelling}}, that is, it contains only input, addition, and multiplication gates.
Many concrete problems like {\sc{$r$-Simple $k$-Path}} can be reduced to \probMonoTest by encoding the set of candidate objects as monomials of some large polynomial,
so that ``good'' objects correspond to monomials with low individual degrees. As we will see in a moment, this is also the case for \probLab.

Let $(G=(V,E),L)$ be an instance of the \probLab problem and let $\Ii$ be the family of all independent sets of $G$.
We denote $n=|V|$. Let $C_{a}(G,L)$ denote the set of all functions $c:V\rightarrow 2^{[a]}$ such that for every edge $uv\in E$ the sets $c(u)$ and $c(v)$ are disjoint, and for every vertex $v$ we have $c(v)\subseteq L(v)$.
Consider the following polynomial in $n(a+1)$ variables $\{x_v\}_{v\in V}$ and $\{y_{v,j}\}_{v\in V, j\in [a]}$, over $\field{2}$.

\begin{equation}
\label{eq:p_G}
 p_G = \sum_{\substack{c\in C_a(G,L)\\\sum_v|c(v)|=bn}}\prod_{v\in V}x_v^{|c(v)|}\prod_{j\in c(v)}y_{v,j}.
\end{equation}

Note that every summand in expression~\eqref{eq:p_G} has a different set of variables, therefore it corresponds to a monomial (with coefficient 1).
Then the following proposition is immediate.

\begin{proposition}
\label{prop:equiv}
 There is a list ($a$:$b$)-coloring of graph $G$ iff $p_G$ contains a $b$-monomial. 
\end{proposition}

Now we show that $p_G$ can be evaluated relatively fast.

\begin{lemma}
\label{lem:circuit-eval}
 The polynomial $p_G$ can be evaluated using a circuit of size $2^n \poly(a,n)$, which can be constructed in time $2^n \poly(a,n)$
\end{lemma}

\begin{proof}
Consider the following polynomial:
\begin{equation}\label{eq:polynom}
 q_G = \prod_{j=1}^a\sum_{I \in \Ii}\prod_{v\in I}x_vy_{v,j}.
\end{equation}
Observe that $p_G$ is obtained from $q_G$ by removing all monomials of degree different than $2bn$. 
Eq.~\eqref{eq:polynom} shows that $q_G$ can be evaluated by a circuit $C_q$ of size $|\Ii|\poly(a,n)\leq 2^n \poly(a,n)$, which can be constructed in time $2^n \poly(a,n)$.
We obtain from $C_q$ a circuit $C_p$ for $p_G$ by splitting gates according to degrees, in a bottom-up fashion, as follows.

Every input gate $u$ of $C_q$ is replaced with a gate $u_1$ in $C_p$.
Every addition gate $u$ with inputs $x$ and $y$ in $C_q$ is replaced in $C_p$ by $2an$ addition gates $u_1,\ldots,u_{2an}$, where $u_i$ has inputs $x_i$ and $y_i$ (whenever $x_i$ and $y_i$ exist).
Every multiplication gate $u$ with inputs $x$ and $y$ in $C_q$ is replaced in $C_p$ by $2an$ addition gates $u_1,\ldots,u_{2an}$. 
Moreover, for every pair of integers $1\leq r,s \leq 2an$ we create a multiplication gate $u_{r,s}$ with inputs $x_r$ and $y_s$ (whenever they exist) and make it an input of the addition gate $u_{r+s}$.
It is easy to see that for every gate $u$ of $C_q$, for every $i$, the gate $u_i$ of $C_p$ evaluates the same polynomial as $u$, but restricted to monomials in which the total degree is equal to $i$.
When $o$ is the output gate of $C_q$, then $o_{2bn}$ is the output gate of $C_p$.
Clearly, $|C_p| \leq (2an+1)^2|C_q|$, and $C_p$ can be constructed from $C_q$ in time $2^n \poly(a,n)$.
\end{proof}

Since $p_G$ is a homogenous polynomial of degree $k=2bn$, by putting $r=b$ we can combine Proposition~\ref{prop:equiv}, Theorem~\ref{thm:abasi} and Lemma~\ref{lem:circuit-eval} 
to get yet another polynomial-space algorithm for \probLab, running in time $b^{\Oh(n)}\cdot \poly(n)$. 
Similarly, if the running time in Theorem~\ref{thm:abasi} was improved from to $2^{o(\log r/r)\cdot k}\cdot |C|\poly(r,n)$, 
then we would get an algorithm for \probLab in time $2^{o(\log b)\cdot n}\cdot \poly(n)$, which contradicts ETH by Theorem~\ref{th:listab}. 
However, a careful examination shows that this chain of reductions would only yield instances of \probMonoTest with $r=\Oh(\sqrt{k/\log k})$.
Hence, this does not exclude the existence of a fast algorithm that works only for large $r$. Below we show a more direct reduction, which excludes fast algorithms for a wider spectrum of pairs $(r,k)$.
                                                                                                                                                
In the \probCLSubsetSum problem, we are given $n+1$ numbers $s,a_1,\ldots,a_n$, each represented as $n$ decimal digits.
For any number $x$, the $j$-th decimal digit of $x$ is denoted by $x^{(j)}$.
It is assumed that $\sum_{i=1}^n a_i^{(j)} < 10$, for every $j=1,\ldots,n$.
The goal is to verify whether there is a sequence of indices $1 \le i_1 < \ldots < i_k \le n$ such that $\sum_{q=1}^k a_{i_q} = s$.
(Note that by the small sum assumption, this is equivalent to the statement that $\sum_{q=1}^k a_{i_q}^{(j)} = s^{(j)}$, for every $j=1,\ldots,n$).
The standard NP-hardness reduction from {\sc{3-SAT}} to {\sc{Subset Sum}} in fact outputs an instance of \probCLSubsetSum of linear size, yielding the following.

\begin{lemma}\label{lem:cless-lb}
Unless ETH fails, there is no algorithm that solves \probCLSubsetSum with $n$ numbers in time $2^{o(n)}$.
\end{lemma}
\begin{proof}
Let $\varphi$ be an instance of {\sc{3-SAT}} with $N$ variables and $M$ clauses.
By a standard NP-hardness reduction for {\sc Subset Sum} (see e.g.\ the textbook of Cormen et al.~\cite{cormen}) 
in polynomial time one can build an equivalent instance of \probCLSubsetSum, with $\Oh(N+M)$ numbers, each having of $\Oh(N+M)$ decimal digits, 
and with the sum of $j$-th digit in all the numbers not exceeding $7$. In case the number of numbers is different from the length of their decimal representations, we can make them equal by padding the instance by zero numbers
or with zeroes in the decimal representations. Thus, by Theorem~\ref{thm:eth-main}, an $2^{o(n)}$ algorithm for \probCLSubsetSum would contradict ETH.
\end{proof}
                   
We proceed to reducing \probCLSubsetSum to \probMonoTest.
Let us choose a parameter $t\in \{1,\ldots,n\}$. 
We assume w.l.o.g.\ that $n \bmod t = 0$, for otherwise we add $t- (n \bmod t)$ zeroes at the end of every input number.
Let $q=n/t$.
For an $n$-digit decimal number $x$, for every $j=1,\ldots t$, let $x^{[j]}$ denote the $q$-digit number given by the $j$-th block of $q$ digits in $x$, i.e., 
\[x^{[j]} = (x^{(jq-1)}\cdots x^{((j-1)q)})_{10}.\]
Let $r = 10^{q}-1$. Define the following polynomial over $\field{2}$:
\[q_S = \prod_{i=1}^n\left(y_i + z_i\cdot \prod_{j=1}^t x_j^{a_i^{[j]}}\right)\cdot \prod_{j=1}^tx_j^{r - s^{[j]}}.\]
        
\begin{proposition}
\label{prop:equiv2}
 $(s,a_1,\ldots,a_n)$ is a YES-instance of \probCLSubsetSum iff $q_S$ contains the monomial $\prod_{j=1}^tx_j^r\prod_{i\not\in S}y_i\prod_{i\in S}z_i$, for some $S\subseteq\{1,\ldots,n\}$. \qed
\end{proposition}

\begin{proof}
 Consider the following polynomial over $\field{2}$: \[r_S = \sum_{S\subseteq\{1,\ldots,n\}}\,\prod_{j=1}^tx_j^{\sum_{i\in S}a_i^{[j]}+r-s^{[j]}}\prod_{i\not\in S}y_i\prod_{i\in S}z_i.\]
 The summands in the expression above have unique sets of $y_i$ variables, so each of them corresponds to a monomial (of coefficient 1).
 It is clear that these monomials where for every $j$ the degree of $x_j$ is exactly $r$ are in one-to-one correspondence with solutions of the instance $(s,a_1,\ldots,a_n)$.
 The claim follows by observing that polynomials $r_S$ and $q_S$ coincide.
\end{proof}

Let $p_S$ denote the polynomial obtained from $q_S$ by filtering out all the monomials of degree different than $k=tr+n$. 
                   
\begin{proposition}
\label{prop:equiv3}
 $(s,a_1,\ldots,a_n)$ is a YES-instance of \probCLSubsetSum iff $p_S$ contains an $r$-monomial. 
\end{proposition}
         
\begin{proof}
 If $(s,a_1,\ldots,a_n)$ is a YES-instance and let then by Proposition~\ref{prop:equiv2} polynomial $q_S$ contains the monomial $\prod_{j=1}^tx_j^r\prod_{i\not\in S}y_i\prod_{i\in S}z_i$, which is an $r$-monomial.
 This monomial has degree $tr+n$, so it is contained in $p_S$ as well.
 
 Conversely, assume $p_S$ contains an $r$-monomial $m$. Every monomial of $q_s$ (and hence also of $p_S$) contains exactly one of the variables $y_i$ and $z_i$, with degree $1$, for every $i=1,\ldots,n$.
 It means that the total degree of $x_j$-type variables in $m$ is $tr$. Hence, since $m$ is an $r$-monomial, each of $x_j$'s has degree exactly $r$. In other words, $m$ is of the form $\prod_{j=1}^tx_j^r\prod_{i\not\in S}y_i\prod_{i\in S}z_i$, for some $S\subseteq\{1,\ldots,n\}$. Then $(s,a_1,\ldots,a_n)$ is a YES-instance of \probCLSubsetSum by Proposition~\ref{prop:equiv2}.
\end{proof}

\begin{proposition}
\label{prop:size}
 $p_S$ can be evaluated by a circuit of size $\Oh(nt^2r + n^2t)$, which can be constructed in time polynomial in $n+t+r$.
\end{proposition}

\begin{proof}
 Polynomial $q_S$ can be evaluated by a circuit of size $\Oh(nt)$. The circuit for $p_S$ is built using the construction from Lemma~\ref{lem:circuit-eval}. Thus, its size is $\Oh(nt(tr+n))=\Oh(nt^2r + n^2t)$. 
\end{proof}

We are ready to give our main lower bound for \probMonoTest. We state it in the most general form, which unfortunately is also quite technical. Next, we derive an exemplary corollary that gives
a lower bound for $r$ expressed as a function of $k$.

 \begin{theorem}
 \label{thm:ss->mono}
 If there is an algorithm solving \probMonoTest in time $2^{o(k\log r/ r)}|C|^{\Oh(1)}$, then ETH fails.
 The statement remains true even if the algorithm works only for instances where $r=2^{\Theta(n/t(n))}$ and $k=t(n) 2^{\Theta(n/t(n))}$, for an arbitrarily chosen function $t:\mathbb{N}\rightarrow\mathbb{N}$ computable in $2^{o(n)}$ time, such that $t(n)=\omega(1)$ and $t(n)\le n$ for every $n$.
\end{theorem}
                   
\begin{proof}
 By Lemma~\ref{lem:cless-lb}, it suffices to give an algorithm for \probCLSubsetSum that works in time $2^{o(n)}$, where $n$ is the number of input numbers.
 Let $t=t(n)$ and $q=n/t$, $r = 10^q - 1$, $k=tr+n$ as before. 
 Note that $r=2^{\Theta(n/t(n))}$.
 Also, since $10^{n/t(n)}=\Omega(n/t(n))$, $k=\Theta(t(n) 10^{n/t(n)}+n) =\Theta(t(n) 10^{n/t(n)})=t(n) 2^{\Theta(n/t(n))}$.

 By Proposition~\ref{prop:equiv3}, solving \probCLSubsetSum is equivalent to detecting an $r$-monomial in $p_S$, which is a homogenous polynomial of degree $k=tr+n$.
 Let $C$ be the circuit for $p_S$; by Proposition~\ref{prop:size} we have $|C|=\Oh(nt^2r + n^2t)$. 
 If this can be done in time $2^{o(k\log r/r)}|C|^{\Oh(1)}$, we get an algorithm for \probCLSubsetSum running in time
 \[2^{o(k\log r/r)}|C|^{\Oh(1)} = 2^{o((tr+n)q/r)}(ntr)^{\Oh(1)}=2^{o(n+nq/10^q)}(ntr)^{\Oh(1)}=2^{o(n)}(ntr)^{\Oh(1)}.\]
 Recall that $t\le n$ and $r=10^{n/t}-1=2^{o(n)}$, since $t=t(n)=\omega(1)$. Hence $(ntr)^{\Oh(1)} = 2^{o(n)}\poly(n)$.
 The claim follows.
\end{proof}

\begin{theorem}\label{thm:main-mono-detection}
Let $\sigma\in [0,1)$.
Then, unless ETH fails, there is no algorithm for \probMonoTest that solves instances with $r=\Theta(k^\sigma)$ in time $2^{o(k\cdot \frac{\log r}{r})}\cdot |C|^{\Oh(1)}$.
\end{theorem}
\begin{proof}
We prove that an algorithm for \probMonoTest with properties as in the statement can be used to derive an algorithm for the same problem with properties as in the statement of Theorem~\ref{thm:ss->mono},
which implies that ETH fails. Take $t$ to be a positive integer not larger than $n$ such that
\begin{equation}\label{eq:choicet}
\frac{1}{2}\leq \frac{10^{n/t}-1}{(t\cdot (10^{n/t}-1)+n)^\sigma}\leq 2;
\end{equation}
it can be easily verified that since $\sigma<1$, for large enough $n$ such an integer $t\le n$ always exists. Moreover, we have that $t=t(n)\in \omega(1)$ and $t(n)$ can be computed in polynomial time by brute-force.
Hence, $t(n)$ satisfies the properties stated in Theorem~\ref{thm:ss->mono}.


Let $t=t(n)$ and $q=n/t$. Define $r = 10^q - 1$ and $k=tr+n$, then~\eqref{eq:choicet} is equivalent to
$$1/2\leq r/k^\sigma\leq 2.$$
Hence $r=\Theta(k^\sigma)$. Consequently, the assumed algorithm solves \probMonoTest in time $2^{o(k\log r/ r)}|C|^{\Oh(1)}$, however in the proof of Theorem~\ref{thm:ss->mono}
we have shown that the existence of an algorithm that achieves such a running time for this particular choice of parameters implies that ETH fails.
\end{proof}

Note that Theorem~\ref{thm:main-mono-detection} in particular implies that \probMonoTest does not admit an algorithm that achieves running time $2^{o(\frac{\log r}{r})\cdot k}\cdot |C|^{\Oh(1)}$ for any given $r$.

\paragraph*{Acknowledgements.} The authors thank Andreas Bj\"orklund and Matthias Mnich for sharing the problem considered in this paper.

\vfill

\bibliographystyle{abbrv}
\bibliography{ab-coloring}

%
%
%
%
%
%

\end{document}